\documentclass{article}
\usepackage{romp}

\usepackage{amsmath,amsthm,amssymb}

\usepackage{graphicx}

\def\im{\mathrm i}
\def\e{\mathrm e}

\title{ Gaps in the spectrum of a cuboidal periodic lattice graph }
\author{ Ond\v{r}ej Turek\thanks{ The author thanks Pavel Exner for valuable comments and discussions. The research was supported by the Czech Science Foundation (GA\v{C}R) within the project 17-01706S. }
                      \\ Department of Mathematics, Faculty of Science, University of Ostrava, \\ 30.~dubna 22, 701 03 Ostrava, Czech Republic, \\[.3em]
											Nuclear Physics Institute, Czech Academy of Sciences, \\ Hlavn\'{i} 130, 250 68 \v{R}e\v{z}, Czech Republic, \\[.3em]
											Laboratory for Unified Quantum Devices, Kochi University of Technology, \\ 185 Miyanokuchi, Tosayamada, Kami-shi, Kochi 782-8502, Japan
											\\ e-mail: ondrej.turek@osu.cz }

\begin{document}

\maketitle
\begin{abstract}
     We locate gaps in the spectrum of a Hamiltonian on a periodic cuboidal (and generally hyperrectangular) lattice graph with $\delta$ couplings in the vertices. We formulate sufficient conditions under which the number of gaps is finite. As the main result, we find a connection between the arrangement of the gaps and the coefficients in a continued fraction associated with the ratio of edge lengths of the lattice. This knowledge enables a straightforward construction of a periodic quantum graph with any required number of spectral gaps and---to some degree---to control their positions; i.e., to partially solve the inverse spectral problem.
\end{abstract}


\section{Introduction}

A quantum graph is a pair $(\Gamma,H)$, where $\Gamma$ is a metric graph and $H$ is a Hamiltonian on $\Gamma$. An intensive study of quantum graphs began in the 1980s along with the technological progress achieved in manufacturing nanosized graph-like objects, for which quantum graphs are suitable models. Besides, quantum graphs are convenient objects for illustrating various quantum effects in a simple setting.
These facts attracted the attention of researchers to the subject and led to its fast development during the last three decades. A detailed exposition of both the theory and applications of quantum graphs, together with an extensive list of relevant literature, can be found in monograph~\cite{BK13}.

One of the most important characteristics of any quantum system is the spectrum of its Hamiltonian, which is often referred to simply as ``spectrum of the system''. The present paper is concerned with spectral properties of periodic quantum graphs.
Not surprisingly, periodic graphs have band--gap spectra as other periodic quantum systems. However, they show an anomaly concerning the number of gaps. The Bethe--Sommerfeld conjecture of 1933 \cite{BS33} says that for every manifold that is periodic in more than one dimension there is a threshold such that all gaps in the spectrum lie below that value; therefore, the number of gaps is bounded. The conjecture was proved for manifolds of dimension $2$ or more \cite{Sk79, Sk85, DT82, HM98, Pa08}, but its statement turned out to be invalid for quantum graphs, which are $1$-dimensional manifolds.\footnote{Let us emphasize that a metric graph, even if being periodic in several dimensions, is a manifold of dimension $1$. For example, the cuboidal lattice graph in Figure~\ref{Mrizka} is a manifold of dimension $1$ that is periodic in three dimensions.} There are numerous counterexamples of graphs that have infinitely many gaps in the spectrum. Such systems can be constructed using the idea of ``graph decorations'', which was first introduced for combinatorial graphs~\cite{SA00} and then extended to metric graphs~\cite{Ku05}. Also, a use of various exotic couplings in graph vertices can generate an infinite series of gaps in the spectrum~\cite{ET10}.

More interestingly, for decades there was not a single known example of a quantum graph that obeys the Bethe--Sommerfeld rule in the sense of having a nonzero finite number of gaps. Indeed, every periodic quantum graph studied in the literature proved to have either infinitely many gaps, or no gaps at all.
The very existence of a periodic quantum graph featuring a nonzero finite number of gaps was an open problem until 2017, when an example of a graph with this property was explicitly constructed \cite{ET17,ET17b}. The graph had a form of a planar rectangular lattice supporting $\delta$ couplings in the vertices, with the edge lengths and coupling strength carefully adjusted. Its construction was inspired by achievements of 1990s~\cite{Ex95,Ex96,EG96}.

In view of the result \cite{ET17}, it is a natural to ask about spectral properties of quantum graphs periodic in three dimensions, such as cuboidal lattices. It is readily seen that there are quantum graphs with $3$-dimensional periodicity that have infinitely many gaps in the spectrum, as well as graphs with no gaps at all (we will provide explicit examples in Sections~\ref{Sect.SC} and \ref{Sect.gaps}, respectively). The interesting case thus concerns the existence of a periodic quantum graph for which the number of gaps in the spectrum is finite but nonzero.

\begin{figure}[h]
\begin{center}
\includegraphics[angle=0,width=0.4\textwidth]{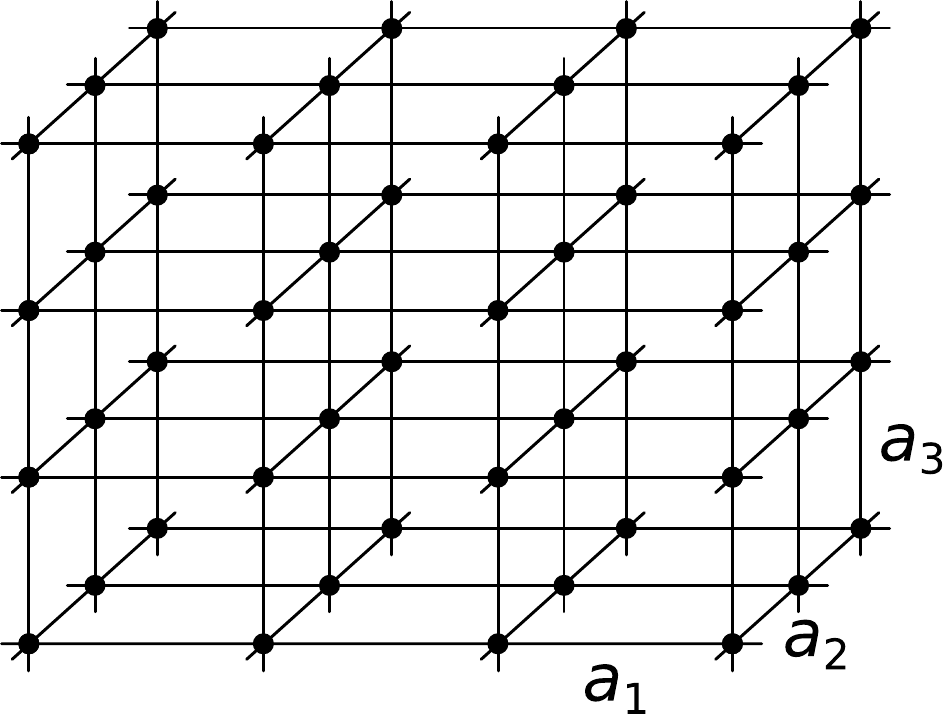}
\caption{Cuboidal lattice graph with edge lengths $a_1,a_2,a_3$.}\label{Mrizka}
\end{center}
\end{figure}

With a view to establish the existence of such graph and to design an explicit example, we will consider a cuboidal periodic lattice, see Figure~\ref{Mrizka}.
The lattice is parametrized by the edge lengths and by the strength of the $\delta$ couplings in the vertices. It will be shown in Section~\ref{Incommensurable} that the number of gaps is closely related to the character of irrationality of edge lengths ratios. The effect of irrationality is studied in more detail in Section~\ref{Section continued fractions}, where a connection between the continued fraction expansion of the edge length ratios and the arrangement of the gaps is revealed. This explicit result enables a solution of the inverse spectral problem in the sense of constructing a periodic quantum graph featuring a prescribed number of spectral gaps along with the possibility to partially control the distances between the gaps. Our construction is the first example to date of a three-dimensional periodic graph with such property.

\section{Spectral condition}\label{Sect.SC}

Consider a periodic lattice with a hyperrectangular cell having edge lengths $a_1,\ldots,a_d$. A choice $d=2$ corresponds to a planar rectangular lattice, $d=3$ corresponds to the cuboidal lattice, depicted in Figure~\ref{Mrizka}. Values $d\geq4$ have no such a simple interpretation, but we will use a parameter $d$ during the calculations instead of a concrete number for the sake of generality. We assume free motion of a particle along the edges and a presence of the $\delta$ couplings of strength $\alpha$ in the vertices.

The aim of this section is to derive the spectral condition for the system, i.e., the equation determining the spectrum. With regard to the periodicity, the Floquet--Bloch decomposition will be applied. Figure~\ref{Bunka} shows the elementary cell of the lattice, together with the notation for eigenvalue components that we will use in the sequel. The figure is plotted for the case $d=3$, but the situation in other dimensions is analogous.
\begin{figure}[h]
\begin{center}
\includegraphics[angle=0,width=0.3\textwidth]{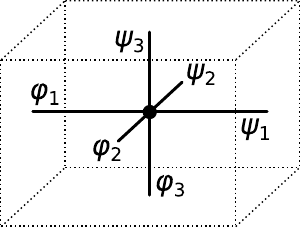}
\caption{The elementary cell of a cuboidal lattice graph. The domain of $\varphi_j$ is $[-a_j/2,0]$, the domain of $\psi_j$ is $[0,a_j/2]$ for every $j$.}\label{Bunka}
\end{center}
\end{figure}

Since no additional potentials on the edges are assumed, the Hamiltonian acts as $\psi\mapsto-\psi''$ on each wavefunction component (let us adhere to the usual convention $\hbar=2m=1$ for simplicity). The Schr\"odinger equation at energy $E=k^2>0$ requires the wavefunction to satisfy $-\psi''=k^2\psi$ on every edge of a graph. Therefore, all the eigenfunction components for a given $E=k^2$ are linear combinations of $\e^{\im kx}$ and $\e^{-\im kx}$. This applies also to the edges constituting the elementary cell, where the wavefunction components can be written as
\begin{equation}\label{vlnfce}
\begin{split}
\psi_j(x)&=C_j^+\e^{\im k x}+C_j^-\e^{-\im k x},\quad x\in[0,a_j/2]\,;\\
\varphi_j(x)&=D_j^+\e^{\im k x}+D_j^-\e^{-\im k x},\quad x\in[-a_j/2,0]
\end{split}
\end{equation}
for $j=1,\ldots,d$.
The $\delta$ coupling with parameter $\alpha\in\mathbb{R}$ in the vertex is represented by boundary conditions
\begin{gather}
\psi_1(0)=\varphi_1(0)=\cdots=\psi_d(0)=\varphi_d(0)\,; \label{delta1}\\
\sum_{j=1}^d(\psi_j'(0)-\varphi_j'(0))=\alpha\cdot\psi_1(0)\,, \label{delta2}
\end{gather}
where the left hand side of equation \eqref{delta2} is the sum of (limits of) derivatives taken in the outgoing sense. Substituting \eqref{vlnfce} into \eqref{delta1} and \eqref{delta2}, one obtains
\begin{gather}
C_1^++C_1^-=D_1^++D_1^-=\cdots=C_d^++C_d^-=D_d^++D_d^-\,; \label{spojitost}\\
\im k\sum_{j=1}^d(C_j^+-C_j^--D_j^++D_j^-)=\alpha(C_1^++C_1^-)\,. \label{sum of der.}
\end{gather}
The Floquet--Bloch decomposition is implemented by imposing conditions
\begin{gather*}
\psi_j(a_j/2)=\e^{\im\theta_j}\varphi_j(-a_j/2)\,;\\
\psi_j'(a_j/2)=\e^{\im\theta_j}\varphi_j'(-a_j/2)
\end{gather*}
for some $\theta_j\in(-\pi,\pi]$, $j=1,\ldots,d$; i.e.,
\begin{gather*}
C_j^+\e^{\im k a_j/2}+C_j^-\e^{-\im k a_j/2}=\e^{\im\theta_j}\left(D_j^+\e^{-\im k a_j/2}+D_j^-\e^{\im k a_j/2}\right); \\
\im k\left(C_j^+\e^{\im k a_j/2}-C_j^-\e^{-\im k a_j/2}\right)=\e^{\im\theta_j}\im k\left(D_j^+\e^{-\im k a_j/2}-D_j^-\e^{\im k a_j/2}\right).
\end{gather*}
Hence
\begin{equation}\label{D pomoci C}
D_j^+=\e^{\im (ka_j-\theta_j)}C_j^+ \qquad\text{and}\qquad
D_j^-=\e^{\im (-ka_j-\theta_j)}C_j^-
\end{equation}
for all $j=1,\ldots,d$. Plugging \eqref{D pomoci C} into equation $C_j^++C_j^-=D_j^++D_j^-$, cf.~\eqref{spojitost}, one gets
\begin{equation}\label{C- pomoci C+}
C_j^-=\frac{\e^{\im (ka_j-\theta_j)}-1}{1-\e^{\im (-ka_j-\theta_j)}}C_j^+\,.
\end{equation}
Now we use once again \eqref{spojitost}, this time the equation $C_j^++C_j^-=C_1^++C_1^-$, combining it with~\eqref{C- pomoci C+} in order to express $C_j^+$ in terms of $C_1^+$. This gives
\begin{equation}\label{C_j pomoci C_1}
C_j^+=\frac{\e^{\im (ka_1-\theta_1)}-\e^{\im (-ka_1-\theta_1)}}{1-\e^{\im (-ka_1-\theta_1)}}\cdot\frac{1-\e^{\im (-ka_j-\theta_j)}}{\e^{\im (ka_j-\theta_j)}-\e^{\im (-ka_j-\theta_j)}}C_1^+
\end{equation}
for $j=2,\ldots,d$.
Equations~\eqref{C_j pomoci C_1} together with \eqref{C- pomoci C+} and \eqref{D pomoci C} allow one to express all coefficients $C_j^+$, $C_j^-$, $D_j^+$, $D_j^-$ in terms of $C_1^+$. When those expressions are substituted into \eqref{sum of der.}, a simple manipulation leads to the following condition:
\begin{equation}\label{sp cond theta}
\sum_{j=1}^d\frac{\cos\theta_j-\cos ka_j}{\sin ka_j}=\frac{\alpha}{2k}\;.
\end{equation}
According to the Floquet--Bloch theory (cf.\ Sect.~4.3.2 and Thm.~4.3.1 in~\cite{BK13}), $k^2>0$ belongs to the spectrum if and only if there exists a $d$-tuple $(\theta_1,\ldots,\theta_d)\in(-\pi,\pi]^d$ such that \eqref{sp cond theta} is satisfied. This is obviously so if and only if
\begin{equation}\label{sp cond}
\sum_{j=1}^d\min_{\theta_j\in(-\pi,\pi]}\frac{\cos\theta_j-\cos ka_j}{\sin ka_j}\leq\frac{\alpha}{2k}\leq\sum_{j=1}^d\max_{\theta_j\in(-\pi,\pi]}\frac{\cos\theta_j-\cos ka_j}{\sin ka_j}\,.
\end{equation}
It is easy to check that for any $x\in\mathbb{R}$,
\begin{align*}
\min_{\theta\in(-\pi,\pi]}\frac{\cos\theta-\cos x}{\sin x}=&-\cot\left(\frac{x}{2}-\left\lfloor\frac{x}{\pi}\right\rfloor\frac{\pi}{2}\right)=-\tan\left(\left\lceil\frac{x}{\pi}\right\rceil\frac{\pi}{2}-\frac{x}{2}\right)\leq0\,; \\
\max_{\theta\in(-\pi,\pi]}\frac{\cos\theta-\cos x}{\sin x}=&\tan\left(\frac{x}{2}-\left\lfloor\frac{x}{\pi}\right\rfloor\frac{\pi}{2}\right)\geq0\,.
\end{align*}
Hence one gets the spectral conditions as follows:
\begin{equation}\label{alpha>0}
\sum_{j=1}^d\tan\left(\frac{ka_j}{2}-\frac{\pi}{2}\left\lfloor\frac{ka_j}{\pi}\right\rfloor\right)\geq\frac{\alpha}{2k} \;\qquad \text{for}\; \alpha>0\,;
\end{equation}
\begin{equation}\label{alpha<0}
\sum_{j=1}^d\tan\left(\frac{\pi}{2}\left\lceil\frac{ka_j}{\pi}\right\rceil-\frac{ka_j}{2}\right)\geq\frac{|\alpha|}{2k} \;\qquad \text{for}\; \alpha<0\,.
\end{equation}
Note that if $\alpha=0$ (Kirchhoff boundary conditions, which represent free motion in the vertices), condition~\eqref{sp cond} is satisfied for every $k>0$. Hence we infer:
\begin{proposition}{Proposition}\label{Prop. alpha=0}
If $\alpha=0$, the spectrum of $H$ has no gaps.
\end{proposition}
For the sake of convenience, let us set
\begin{align}
F(k):=&2k\sum_{j=1}^d\tan\left(\frac{\pi}{2}\left(\frac{ka_j}{\pi}-\left\lfloor\frac{ka_j}{\pi}\right\rfloor\right)\right); \label{F} \\
G(k):=&2k\sum_{j=1}^d\tan\left(\frac{\pi}{2}\left(\left\lceil\frac{ka_j}{\pi}\right\rceil-\frac{ka_j}{\pi}\right)\right). \label{G}
\end{align}
Using~\eqref{F} and \eqref{G}, one can write the spectral conditions~\eqref{alpha>0} and \eqref{alpha<0}, respectively, in the way
\begin{align}
&k^2\in\sigma(H)\quad\Leftrightarrow\quad
F(k)\geq\alpha \qquad\;\; \text{if $\alpha>0$}\,; \label{Band cond.F}\\
&k^2\in\sigma(H)\quad\Leftrightarrow\quad
G(k)\geq|\alpha| \qquad \text{if $\alpha<0$}\,. \label{Band cond.G}
\end{align}

\section{Gaps in the spectrum}\label{Sect.gaps}

With regard to the spectral conditions~\eqref{Band cond.F} and \eqref{Band cond.G}, the gaps in the spectrum of $H$ are determined by the condition
\begin{align}
&k^2\notin\sigma(H)\quad\Leftrightarrow\quad
F(k)<\alpha \qquad\;\; \text{if $\alpha>0$}\,; \label{Gap cond.F}\\
&k^2\notin\sigma(H)\quad\Leftrightarrow\quad
G(k)<|\alpha| \qquad \text{if $\alpha<0$}\,, \label{Gap cond.G}
\end{align}
where $F(k)$ and $G(k)$ are functions introduced in equations~\eqref{F} and \eqref{G}, respectively.
Let us start from analyzing the function $F(k)$.

\begin{lemma}{Lemma}\label{Lemma F}
Function $F(k)$ on $(0,+\infty)$ has the following properties:
\begin{itemize}
\item[(i)] A $k>0$ is a discontinuity of $F$ if and only if $k=m\pi/a_j$ for some $a_j\in\{a_1,\ldots,a_d\}$, $m\in\mathbb{N}$.
\item[(ii)] $F(k)$ is strictly increasing in each interval of continuity.
\item[(iii)] For each $m\in\mathbb{N}$ and $a_j\in\{a_1,\ldots,a_d\}$,
$$
\lim_{k\nearrow\frac{m\pi}{a_j}}F(k)=+\infty\,.
$$
\end{itemize}
\end{lemma}

\begin{proof}
Function $F$ is a sum of terms $\frac{2\pi}{a_j}x\tan\left(\frac{\pi}{2}(x-\lfloor x\rfloor)\right)$ with $x>0$ taking values $ka_j/\pi$ for $j=1,\ldots,d$. They are discontinuous exactly at integer values of $x$, i.e., at $ka_j/\pi=m\in\mathbb{N}$. This proves statement (i). Since a term $x\tan\left(\frac{\pi}{2}(x-\lfloor x\rfloor)\right)$ obviously increases on each its interval of continuity, the same holds true for $F$, which is a linear combination of such terms with positive coefficients; hence we get statement (ii). Finally, we have
$\tan\left(\frac{\pi}{2}(x-\lfloor x\rfloor)\right)\geq0$ for all $x>0$ and
$$
\lim_{k\nearrow\frac{m\pi}{a_j}}\tan\left(\frac{\pi}{2}\left(\frac{ka_j}{\pi}-\left\lfloor\frac{ka_j}{\pi}\right\rfloor\right)\right)=+\infty\,,\;
$$
hence statement (iii) follows immediately.
\end{proof}

In the same manner, one can demonstrate similar features of the function $G(k)$:
\begin{lemma}{Lemma}\label{Lemma G}
Function $G(k)$ on $(0,+\infty)$ has the following properties:
\begin{itemize}
\item[(i)] A $k>0$ is a discontinuity of $G$ if and only if $k=m\pi/a_j$ for some $a_j\in\{a_1,\ldots,a_d\}$, $m\in\mathbb{N}$.
\item[(ii)] $G(k)$ is strictly decreasing in each interval of continuity.
\item[(iii)] For each $m\in\mathbb{N}$ and $a_j\in\{a_1,\ldots,a_d\}$,
$$
\lim_{k\searrow\frac{m\pi}{a_j}}G(k)=+\infty\,.
$$
\end{itemize}
\end{lemma}

Theorems~\ref{Thm. gaps F} and \ref{Thm. gaps G} below characterize the positions of gaps for $\alpha>0$ and $\alpha<0$, respectively. Since their proofs are similar to each other, we will carry out only one of them.

\begin{theorem}{Theorem}\label{Thm. gaps F}
Let $\alpha>0$.
\begin{itemize}
\item Each spectral gap has the left (lower) endpoint equal to $k^2=\left(\frac{m\pi}{a_j}\right)^2$ for some $j\in\{1,\ldots,d\}$ and some $m\in\mathbb{N}$.
\item A gap adjacent to $k^2=\left(\frac{m\pi}{a_{\ell}}\right)^2$ is present if and only if
\begin{equation}\label{gap F}
\frac{2m\pi}{a_{\ell}}\sum_{j=1}^d\tan\left(\frac{\pi}{2}\left(m\frac{a_j}{a_{\ell}}-\left\lfloor m\frac{a_j}{a_{\ell}}\right\rfloor\right)\right)<\alpha.
\end{equation}
\end{itemize}
\end{theorem}

\begin{proof}
Lemma~\ref{Lemma F} implies that in each interval of continuity of $F$, the values $F(k)$ grow up to infinity. Therefore, the values $k$ obeying the gap condition~\eqref{Gap cond.F} form at most one connected set in each interval of continuity of $F$. This connected set is adjacent to the left endpoint of the interval of continuity and does not extend to its right endpoint. Consequently, the left endpoint of any gap is $k^2=(m\pi/a_\ell)^2$ for some $\ell\in\{1,\ldots,d\}$ and $m\in\mathbb{N}$.

Moreover, a gap with the left endpoint equal to $k^2=(m\pi/a_{\ell})^2$ is present if and only if $\lim_{k\searrow\frac{m\pi}{a_{\ell}}}F(k)<\alpha$. Since
$$
\lim_{k\searrow\frac{m\pi}{a_{\ell}}}F(k)=\frac{2m\pi}{a_{\ell}}\sum_{j=1}^d\tan\left(\frac{\pi}{2}\left(m\frac{a_j}{a_{\ell}}-\left\lfloor m\frac{a_j}{a_{\ell}}\right\rfloor\right)\right),
$$
the gap condition~\eqref{gap F} follows immediately.
\end{proof}

\begin{theorem}{Theorem}\label{Thm. gaps G}
Let $\alpha<0$.
\begin{itemize}
\item Each spectral gap has the right (upper) endpoint equal to $k^2=\left(\frac{m\pi}{a_j}\right)^2$ for some $j\in\{1,\ldots,d\}$ and $m\in\mathbb{N}$.
\item A gap adjacent to $k^2=\left(\frac{m\pi}{a_{\ell}}\right)^2$ is present if and only if
\begin{equation}\label{gap G}
\frac{2m\pi}{a_{\ell}}\sum_{j=1}^d\tan\left(\frac{\pi}{2}\left(\left\lceil m\frac{a_j}{a_{\ell}}\right\rceil-m\frac{a_j}{a_{\ell}}\right)\right)<|\alpha|.
\end{equation}
\end{itemize}
\end{theorem}
As a corollary of Theorems~\ref{Thm. gaps F} and \ref{Thm. gaps G}, we obtain an example of a quantum graph with infinitely many gaps in the spectrum:

\begin{corollary}{Corollary}\label{Coro. infinite}
If $\alpha\neq0$ and all the edge lengths are commensurable, i.e., $a_j/a_\ell$ is a rational number for all $j,\ell\in\{1,\ldots,d\}$, then the number of gaps in the spectrum of $H$ is infinite.
\end{corollary}

\begin{proof}
According to the assumptions, there exist $r_2,\ldots,r_d\in\mathbb{N}$ and $s_2,\ldots,s_d\in\mathbb{N}$ such that
$$
\frac{a_j}{a_1}=\frac{r_j}{s_j}\in\mathbb{Q}\,, \qquad j=2,\ldots,d\,.
$$
Let $s$ be the least common multiple of $s_2,\ldots,s_d$; then $sa_j/a_1=r_j\cdot s/s_j\in\mathbb{N}$ for all $j=2,\ldots,d$. Then every $h\in\mathbb{N}$ satisfies
$$
\left\lfloor hs\frac{a_j}{a_1}\right\rfloor=\left\lceil hs\frac{a_j}{a_1}\right\rceil=hs\frac{a_j}{a_1}\in\mathbb{N} \qquad \text{for all $j=2,\ldots,d$}.
$$
Consequently, the left hand side of condition~\eqref{gap F} (case $\alpha>0$) and of condition~\eqref{gap G} (case $\alpha<0$) with $\ell=1$ and $m=hs$ vanishes for every $h\in\mathbb{N}$. This gives rise to infinitely many spectral gaps adjacent to points $(hs\pi/a_1)^2$ for $h\in\mathbb{N}$. Notice that the statement holds true for either sign of $\alpha$.
\end{proof}

\section{Finite number of spectral gaps}\label{Incommensurable}

In view of the Bethe--Sommerfeld conjecture, the most interesting situation occurs when the number of spectral gaps is finite. Let us recall that we have already found a necessary condition for the finiteness of the number of gaps in Section~\ref{Sect.gaps}: putting aside the case $\alpha=0$ (when the spectrum has trivially no gaps, see Proposition~\ref{Prop. alpha=0}), the ratios of some edge lengths must be irrational (Corollary~\ref{Coro. infinite}). The aim of this section is to find sufficient conditions.

To that end, we will need an appropriate characterization of irrationality. It is convenient to start from a classical notion of \emph{Markov constant}. The word ``constant'' is somewhat misleading, because the Markov constant is actually defined as a function of a real argument $\gamma\in\mathbb{R}$ as follows:
\begin{equation}\label{mu(theta)}
\mu(\gamma)=\inf\left\{c>0\;\left|\;\left(\exists_\infty(p,q)\in\mathbb{Z}^2\right)\left(\left|\gamma-\frac{p}{q}\right|<\frac{c}{q^2}\right)\right.\right\}.
\end{equation}
An (irrational) $\gamma\in\mathbb{R}$ is called \emph{badly approximable} \cite[p.~22]{Sch80} if there exists a $c>0$ such that $|\gamma-p/q|>c/q^2$ for all rational $p/q$. In other words, $\gamma$ is badly approximable if $\mu(\gamma)>0$. Badly approximable numbers form an uncountable set of Lebesgue measure zero in $\mathbb{R}$; see \cite[Theorem~5F and Corollary~5G]{Sch80} and \cite[Theorem~29]{Kh64}.

We aim at formulating a sufficient condition for the finiteness of the number of gaps in terms of a function that is defined similarly to \eqref{mu(theta)}, but takes $\gamma>0$ and does not use the absolute value; namely
\begin{equation}\label{upsilon}
\upsilon(\gamma)=\inf\left\{c>0\;\left|\;\left(\exists_\infty(p,q)\in\mathbb{Z}^2\right)\left(0<\gamma-\frac{p}{q}<\frac{c}{q^2}\right)\right.\right\}.
\end{equation}
Function $\upsilon(\gamma)$ can be viewed as ``finer'' than Markov constant in the sense that it distinguishes between $\gamma$ and $\gamma^{-1}$. While the Markov constant satisfies $\mu(\gamma)=\mu(\gamma^{-1})$ for any $\gamma\neq0$, values $\upsilon(\gamma)$ and $\upsilon(\gamma^{-1})$ can differ. The smaller one of them coincides with the Markov constant,
\begin{equation}\label{upsilon iii}
\mu(\gamma)=\min\{\upsilon(\gamma),\upsilon(\gamma^{-1})\}.
\end{equation}
Formula~\eqref{upsilon iii} was derived in~\cite{ET17}, where function $\upsilon(\gamma)$ (``upsilon'') was originally introduced. Ibidem the following basic properties were derived:
\begin{align}
\upsilon(\gamma)=&\inf\left\{c>0\;\left|\;\left(\exists_\infty m\in\mathbb{N}\right)\left(m(m\gamma-\lfloor m\gamma\rfloor)<c\right)\right.\right\}, \label{upsilon i} \\
\upsilon(\gamma^{-1})=&\inf\left\{c>0\;\left|\;\left(\exists_\infty m\in\mathbb{N}\right)\left(m(\lceil m\gamma\rceil-m\gamma)<c\right)\right.\right\}. \label{upsilon ii}
\end{align}
(symbols $\lfloor\cdot\rfloor$ and $\lceil\cdot\rceil$ stand for the floor and the ceiling function, respectively).

With equations~\eqref{upsilon i} and \eqref{upsilon ii} in hand, we are ready to formulate a condition that guarantees the number of gaps in the spectrum of the lattice graph to be finite. Since the condition depends on the sign of the coupling parameter $\alpha$, we distinguish the case $\alpha>0$ and $\alpha<0$ in Proposition~\ref{Prop. finite F} and \ref{Prop. finite G}, respectively.

\begin{proposition}{Proposition}\label{Prop. finite F}
Let $\alpha>0$. If
\begin{equation}\label{alpha>0 fin.}
\alpha<\frac{\pi^2}{a_\ell}\sum_{j=1}^d\upsilon\left(\frac{a_j}{a_{\ell}}\right) \qquad\text{for all $\ell\in\{1,\ldots,d\}$}\,,
\end{equation}
there are at most finitely many gaps in the spectrum of $H$.
\end{proposition}

\begin{proof}
According to Theorem~\ref{Thm. gaps F}, any gap is adjacent to a point $(m\pi/a_\ell)^2$ for some $\ell\in\{1,\ldots,d\}$, $m\in\mathbb{N}$. Let $\ell$ be fixed. If there is a gap adjacent to $(m\pi/a_{\ell})^2$, then $m$ satisfies the gap condition~\eqref{gap F}, i.e.,
$$
\frac{2m\pi}{a_{\ell}}\sum_{j=1}^d\tan\left(\frac{\pi}{2}\left(m\frac{a_j}{a_{\ell}}-\left\lfloor m\frac{a_j}{a_{\ell}}\right\rfloor\right)\right)<\alpha.
$$
The assumption~\eqref{alpha>0 fin.} guarantees the existence of a $c>0$ such that
$$
\alpha=\frac{\pi^2}{a_\ell}\sum_{j=1}^d\upsilon\left(\frac{a_j}{a_{\ell}}\right)-c.
$$
Hence
$$
\frac{2m\pi}{a_{\ell}}\sum_{j=1}^d\tan\left(\frac{\pi}{2}\left(m\frac{a_j}{a_{\ell}}-\left\lfloor m\frac{a_j}{a_{\ell}}\right\rfloor\right)\right)<\frac{\pi^2}{a_\ell}\sum_{j=1}^d\upsilon\left(\frac{a_j}{a_{\ell}}\right)-c,
$$
which is equivalent to
$$
\sum_{j=1}^d\left[\frac{2m\pi}{a_{\ell}}\tan\left(\frac{\pi}{2}\left(m\frac{a_j}{a_{\ell}}-\left\lfloor m\frac{a_j}{a_{\ell}}\right\rfloor\right)\right)-\frac{\pi^2}{a_\ell}\upsilon\left(\frac{a_j}{a_{\ell}}\right)+\frac{c}{d}\right]<0.
$$
Consequently, there obviously exists a $j$ such that
$$
\frac{2m\pi}{a_{\ell}}\tan\left(\frac{\pi}{2}\left(m\frac{a_j}{a_{\ell}}-\left\lfloor m\frac{a_j}{a_{\ell}}\right\rfloor\right)\right)<\frac{\pi^2}{a_\ell}\upsilon\left(\frac{a_j}{a_{\ell}}\right)-\frac{c}{d}\,.
$$
Hence we get, using the estimate $x\leq\tan x$ (valid for all $x\geq0$),
$$
\frac{2m\pi}{a_{\ell}}\cdot\frac{\pi}{2}\left(m\frac{a_j}{a_{\ell}}-\left\lfloor m\frac{a_j}{a_{\ell}}\right\rfloor\right)<\frac{\pi^2}{a_\ell}\upsilon\left(\frac{a_j}{a_{\ell}}\right)-\frac{c}{d}\,,
$$
i.e.,
\begin{equation}\label{ineq prop F}
m\left(m\frac{a_j}{a_{\ell}}-\left\lfloor m\frac{a_j}{a_{\ell}}\right\rfloor\right)<\upsilon\left(\frac{a_j}{a_{\ell}}\right)-\frac{a_\ell}{d\pi^2}c\,.
\end{equation}
Since the right hand side of \eqref{ineq prop F} is strictly less than $\upsilon(a_j/a_\ell)$, inequality~\eqref{ineq prop F} can be satisfied for at most finitely many values of $m$ due to property~\eqref{upsilon i} of $\upsilon$. Consequently, for each $\ell\in\{1,\ldots,d\}$ there are at most finitely many gaps adjacent to points $(m\pi/a_{\ell})^2$. The total number of gaps in the spectrum is thus finite as well.
\end{proof}

\begin{remark}{Remark}
Let us briefly explain why the statement of Proposition~\ref{Prop. finite F} (as well as of Proposition~\ref{Prop. finite G} below) is not formulated as an equivalence. Assume $\alpha>\frac{\pi^2}{a_\ell}\sum_{j=1}^d\upsilon\left(\frac{a_j}{a_{\ell}}\right)$ (here we intentionally take strict inequality $>$ instead of $\geq$) for some $\ell\in\{1,\ldots,d\}$, i.e., $\alpha-\frac{\pi^2}{a_\ell}\sum_{j=1}^d\upsilon\left(\frac{a_j}{a_{\ell}}\right)=c>0$. Using techniques similar to the proof of Proposition~\ref{Prop. finite F}, one can demonstrate that there are infinitely many $m\in\mathbb{N}$ with the property
\begin{equation}\label{ineq impossible}
\frac{2m\pi}{a_{\ell}}\tan\left(\frac{\pi}{2}\left(m\frac{a_j}{a_{\ell}}-\left\lfloor m\frac{a_j}{a_{\ell}}\right\rfloor\right)\right)<\frac{\pi^2}{a_\ell}\upsilon\left(\frac{a_j}{a_{\ell}}\right)+\frac{c}{d}\,.
\end{equation}
The set of numbers $m\in\mathbb{N}$ obeying~\eqref{ineq impossible} depends on $j$. If the intersection of those sets (for all $j$) has infinite cardinality, one can sum \eqref{ineq impossible} over $j$, which leads to the gap condition~\eqref{gap F} valid for an infinite number of integers $m$. So there are infinitely many spectral gaps, which is the desired result. But if the intersection happens to contain only finitely many $m\in\mathbb{N}$, then the existence of infinitely many spectral gaps is not established by this approach.
The only exception is the case $d=2$, where the sum over $d$ consists of only one nonzero term (notice that the term for $j=\ell$ always vanishes). In this case the assumption $\alpha>\frac{\pi^2}{a_\ell}\sum_{j=1}^2\upsilon\left(\frac{a_j}{a_{\ell}}\right)$ for $\ell=1$ or $\ell=2$ implies that the number of spectral gaps is infinite; cf.~\cite[Prop.~4.4]{ET17}.
\end{remark}

\begin{proposition}{Proposition}\label{Prop. finite G}
Let $\alpha<0$. If
\begin{equation}\label{alpha<0 fin.}
|\alpha|<\frac{\pi^2}{a_\ell}\sum_{j=1}^d\upsilon\left(\frac{a_{\ell}}{a_j}\right) \qquad\text{for all $\ell\in\{1,\ldots,d\}$}\,,
\end{equation}
there are at most finitely many gaps in the spectrum of $H$.
\end{proposition}

\begin{proof}
The statement is demonstrated in a manner similar to Proposition~\ref{Prop. finite F}. Starting from the spectral condition~\eqref{gap G} and using the assumption~\eqref{alpha<0 fin.}, one arrives at the condition
\begin{equation}\label{ineq prop G}
m\left(\left\lceil m\frac{a_j}{a_{\ell}}\right\rceil-m\frac{a_j}{a_{\ell}}\right)<\upsilon\left(\frac{a_{\ell}}{a_j}\right)-\frac{a_\ell}{d\pi^2}c\,,
\end{equation}
which plays the role of condition~\eqref{ineq prop F} in the proof of Proposition~\ref{Prop. finite F}.
Since the right hand side of~\eqref{ineq prop G} is less than $\upsilon(a_{\ell}/a_j)$, one can use~\eqref{upsilon ii} with $\gamma=a_j/a_{\ell}$ to infer that for any $\ell\in\{1,\ldots,d\}$ there at most finitely many values of $m$ obeying inequality~\eqref{ineq prop G}. Consequently, the total number of gaps in the spectrum is finite.
\end{proof}

We conclude this section by a common corollary of Propositions~\ref{Prop. finite F} and \ref{Prop. finite G}. It says that under certain assumption on the edge lengths of the lattice, the number of gaps in the spectrum becomes at most finite when the $\delta$ coupling (repulsive or attractive) is weak enough.

\begin{corollary}{Corollary}\label{Coro. finite}
If for each $\ell\in\{1,\ldots,d\}$ there is a $j_\ell\in\{1,\ldots,d\}$ such that $a_{j_\ell}/a_\ell$ is a badly approximable number, then for a sufficiently small $|\alpha|>0$, there are at most finitely many gaps in the spectrum.
\end{corollary}

\begin{proof}
Let $\alpha>0$.
If $a_{j_\ell}/a_\ell$ is badly approximable, then $\upsilon(a_{j_\ell}/a_\ell)\geq\mu(a_{j_\ell}/a_\ell)>0$ (cf.~\eqref{upsilon iii}); hence
\begin{equation}\label{c l}
\sum_{j=1}^d\upsilon\left(\frac{a_j}{a_{\ell}}\right)\geq\upsilon\left(\frac{a_{j_\ell}}{a_\ell}\right)>0\,.
\end{equation}
By assumption, estimate~\eqref{c l} can be made for each $\ell\in\{1,\ldots,d\}$. Consequently,
$$
\min_{\ell\in\{1,\ldots,d\}}\frac{\pi^2}{a_\ell}\sum_{j=1}^d\upsilon\left(\frac{a_j}{a_{\ell}}\right)\geq\min_{\ell\in\{1,\ldots,d\}}\frac{\pi^2}{a_\ell}\upsilon\left(\frac{a_{j_\ell}}{a_\ell}\right)=:c>0.
$$
Then Proposition~\ref{Prop. finite F} guarantees that the number of spectral gaps is at most finite for any $\alpha\in(0,c)$.
Case $\alpha<0$ is treated similarly.
\end{proof}

\section{Continued fractions and inverse spectral problem}\label{Section continued fractions}

In this section we examine the irrationality in light of continued fractions, which proves to be a very fruitful approach. It will show that the coefficients in the continued fraction expansion associated to the edge lengths ratio can directly govern the arrangement of spectral gaps.

Recall that a continued fraction~\cite{Kh64} associated to a $\gamma\in\mathbb{R}$ is a representation of the form
\begin{equation}\label{cont.fr.}
\gamma=c_0+\frac{1}{c_1+\frac{1}{c_2+\frac{1}{c_3+\frac{1}{\cdots}}}}\;,
\end{equation}
where $c_0\in\mathbb{Z}$ and $c_j\in\mathbb{N}$ for all $j>0$. Representation~\eqref{cont.fr.} is usually written in a compact form $\gamma=[c_0;c_1,c_2,c_3,c_4,c_5,\ldots]$.
The expansion is finite if and only if $\gamma\in\mathbb{Q}$. Infinite expansions $[c_0;c_1,c_2,c_3,\ldots]$ always converge, which is guaranteed by \cite[Thm.~10]{Kh64}; thus every sequence $\{c_j\}_{n=0}^\infty$ satisfying $c_0\in\mathbb{Z}$ and $c_j\in\mathbb{N}$ for $j>0$ defines a unique real number $\gamma$.
When an infinite expansion~\eqref{cont.fr.} is terminated at the $n$-th position, one gets a rational number of the form $\frac{p_n}{q_n}=[c_0;c_1,c_2,\ldots,c_n]$, which is called the \emph{$n$-th convergent} of the number $\gamma=[c_0;c_1,c_2,c_3,c_4,c_5,\ldots]$. The convergents satisfy
\begin{equation}\label{convergents}
\frac{p_0}{q_0}<\frac{p_2}{q_2}<\frac{p_4}{q_4}<\frac{p_6}{q_6}<\cdots\leq\gamma\leq\cdots<\frac{p_7}{q_7}<\frac{p_5}{q_5}<\frac{p_3}{q_3}<\frac{p_1}{q_1}\;.
\end{equation}

The method we are going to develop in this section is simple in principle, but requires a technical treatment that depends on actual values of the parameters. Therefore, in order to keep the presentation as straightforward as possible, we will focus on the case $\alpha>0$ and reduce the number of parameters by setting the problem more explicit. Namely, we assume that all edge lengths $a_2,\ldots,a_d$ are equal, that is,
$$
a_1=\gamma a, \quad\text{and}\quad a_j=a \quad \text{for all } j=2,\ldots,d,
$$
where the ratio $\gamma=a_1/a_2$ is irrational.
Moreover, we focus only on those $\gamma$ that satisfy $\gamma\in(0,1)$ (i.e., one has $c_0=0$ in~\eqref{cont.fr.}) and the terms $c_j$ ($j\in\mathbb{N}$) in their continued fraction representation \eqref{cont.fr.} attain only two values, say $1$ and $2$, in the following manner:
\begin{equation}\label{gamma}
\begin{array}{ll}
c_{n}=1 & \text{if $n$ is even or $n=1$}; \\
c_{n}\in\{1,2\} & \text{if $n$ is odd}.
\end{array}
\end{equation}
Our goal is to show that the number and arrangement of $2$'s in the continued fraction expansion of $\gamma$ determines the number and positions of gaps in the spectrum of the lattice.

Before proceeding further, let us state a criterion for comparing continued fractions, which will be useful in the sequel. Since its validity is easy to see, the proof is omitted.
\begin{proposition}{Proposition}\label{Prop. ineq}
Let $\beta=[b_0;b_1,b_2,b_3,b_4,b_5,\ldots]$, $\gamma=[c_0;c_1,c_2,c_3,c_4,c_5,\ldots]$ and $j$ be the minimal index such that $b_j\neq c_j$. Then
$$
\beta<\gamma \qquad\Leftrightarrow\qquad (\text{$j$ is even and $b_j<c_j$}) \quad\text{or}\quad (\text{$j$ is odd and $b_j>c_j$}).
$$
If $\beta=[c_0;c_1,c_2,\ldots,c_m]$ and $\gamma=[c_0;c_1,c_2,c_3,c_4,c_5,\ldots]$, then $\beta<\gamma$ if and only if $m$ is even.
\end{proposition}

Our analysis will also require estimates of the quantities of type $q(q\gamma-p)$ for $p,q\in\mathbb{N}$:

\begin{proposition}{Proposition}\label{Prop. error}
(i)\; If $p_n/q_n$ is a convergent of $\gamma=[c_0;c_1,c_2,c_3,c_4,c_5,\ldots]$, then
$$
q_n|q_n\gamma-p_n|>\frac{1}{c_{n+1}+\frac{1}{c_{n+2}+\frac{1}{c_{n+3}+1}}+\frac{1}{c_{n}+\frac{1}{c_{n-1}+1}}}>\frac{1}{c_{n+1}+\frac{1}{c_{n+2}}+\frac{1}{c_{n}}}
$$
and
$$
q_n|q_n\gamma-p_n|<\frac{1}{c_{n+1}+\frac{1}{c_{n+2}+1}+\frac{1}{c_{n}+1}}\;.
$$
(ii)\; If $p/q$ lies between convergents $p_{n-1}/q_{n-1}$ and $p_{n+1}/q_{n+1}$ of $\gamma$, then
$$
q|q\gamma-p|>\frac{1}{c_{n+1}}\;.
$$
\end{proposition}

\begin{proof}
In order to prove (i), we start from the general formula \cite[Eq.~(2.2)]{Ha16}
\begin{equation}\label{Hancl}
\gamma-\frac{p_n}{q_n}=\frac{(-1)^n}{q_n^2\left([c_{n+1};c_{n+2},\ldots]+[0;c_n,c_{n-1},\ldots,c_1]\right)}\;,
\end{equation}
which follows immediately from~\cite[p.~10]{Sch80}.
Equation~\eqref{Hancl} together with Proposition~\ref{Prop. ineq} applied on the continued fractions in the denominator of~\eqref{Hancl} leads to the sought estimates
$$
q_n|q_n\gamma-p_n|>\frac{1}{[c_{n+1};c_{n+2},c_{n+3}+1]+[0;c_n,c_{n-1}+1]}>\frac{1}{[c_{n+1};c_{n+2}]+[0;c_n]}
$$
and
$$
q_n|q_n\gamma-p_n|<\frac{1}{[c_{n+1};c_{n+2}+1]+[0;c_n+1]}\;.
$$
Let us proceed to the proof of (ii).
By assumption, $p/q$ lies between $p_{n-1}/q_{n-1}$ and $p_{n+1}/q_{n+1}$; thus \eqref{convergents} implies
\begin{equation}\label{gamma est.}
\left|\gamma-\frac{p}{q}\right|\geq\left|\frac{p_{n+1}}{q_{n+1}}-\frac{p}{q}\right|=\frac{|qp_{n+1}-pq_{n+1}|}{q\cdot q_{n+1}}\;.
\end{equation}
Since $|qp_{n+1}-pq_{n+1}|$ is a nonzero integer, its value is greater or equal to $1$. Hence~\eqref{gamma est.} gives
\begin{equation}\label{est.}
q|q\gamma-p|\geq\frac{q}{q_{n+1}}\;.
\end{equation}
In view of estimating $\frac{q}{q_{n+1}}$, we will examine the term $\left|\frac{p}{q}-\frac{p_{n-1}}{q_{n-1}}\right|$ in two ways. At first, one has
\begin{equation}\label{lower est.}
\left|\frac{p}{q}-\frac{p_{n-1}}{q_{n-1}}\right|=\frac{|pq_{n-1}-qp_{n-1}|}{q\cdot q_{n-1}}\geq\frac{1}{q\cdot q_{n-1}}\;,
\end{equation}
which is again obtained using the trivial estimate $|pq_{n-1}-qp_{n-1}|\geq1$. Secondly, since $p/q$ lies between $p_{n-1}/q_{n-1}$ and $p_{n+1}/q_{n+1}$, one gets
\begin{equation}\label{upper est.}
\left|\frac{p}{q}-\frac{p_{n-1}}{q_{n-1}}\right|<\left|\frac{p_{n+1}}{q_{n+1}}-\frac{p_{n-1}}{q_{n-1}}\right|=\frac{c_{n+1}}{q_{n+1}q_{n-1}}\;,
\end{equation}
where the last equality holds due to a known formula $\frac{p_{k-2}}{q_{k-2}}-\frac{p_k}{q_k}=\frac{(-1)^{k-1}c_k}{q_kq_{k-2}}$, see~\cite[eq.~(10)]{Kh64}.
Inequalities \eqref{lower est.} and \eqref{upper est.} together imply
$$
\frac{q}{q_{n+1}}>\frac{1}{c_{n+1}}\;.
$$
Plugging this estimate in~\eqref{est.}, one gets the sought inequality $q|q\gamma-p|>1/c_{n+1}$.
\end{proof}

The presence of spectral gaps is determined by condition~\eqref{gap F}, which contains terms $m\tan\left(\frac{\pi}{2}(m\gamma-\lfloor m\gamma\rfloor)\right)$ with $m\in\mathbb{N}$. The following technical lemma provides their estimates for our particular choice of $\gamma$, cf.~\eqref{gamma}.

\begin{lemma}{Lemma}\label{Levels beta}
Let $\gamma=[0;1,1,c_3,1,c_5,1,c_7,1,\ldots]$, where $c_{2n+1}\in\{1,2\}$ for all $n\in\mathbb{N}$.
For every $m\in\mathbb{N}$, we have:
\begin{itemize}
\item[(i)] If $\lfloor m\gamma\rfloor/m<\gamma$ is not a convergent of $\gamma$, then $\frac{2}{\pi}m\tan\left(\frac{\pi}{2}(m\gamma-\lfloor m\gamma\rfloor)\right)>1$.
\item[(ii)] If $\lfloor m\gamma\rfloor/m$ is the $2n$-th convergent of $\gamma$ for some $n$ (i.e., $\lfloor m\gamma\rfloor=p_{2n}$, $m=q_{2n}$) and $c_{2n+1}=1$, then $\frac{2}{\pi}m\tan\left(\frac{\pi}{2}(m\gamma-\lfloor m\gamma\rfloor)\right)>\frac{2}{5}$.
\item[(iii)] If $\lfloor m\gamma\rfloor/m$ is the $2n$-th convergent of $\gamma$ for some $n$ and $c_{2n+1}=2$, then $\frac{2}{\pi}m\tan\left(\frac{\pi}{2}(m\gamma-\lfloor q_{2n}\gamma\rfloor)\right)<\frac{4}{\pi}(2-\sqrt{3})\approx0.341$.
\end{itemize}
\end{lemma}

\begin{proof}
(i)\; A trivial estimate $\tan x\geq x$, which is valid for any $x\geq0$, gives
\begin{equation}\label{tan estimate}
\frac{2}{\pi}m\tan\left(\frac{\pi}{2}(m\gamma-\lfloor m\gamma\rfloor)\right)\geq\frac{2}{\pi}m\cdot\frac{\pi}{2}(m\gamma-\lfloor m\gamma\rfloor)=m(m\gamma-\lfloor m\gamma\rfloor).
\end{equation}
If $\lfloor m\gamma\rfloor/m$ is not a convergent of $\gamma$, then either $\lfloor m\gamma\rfloor/m$ lies between two convergents of $\gamma$ that are smaller than $\gamma$, i.e.,
\begin{equation}\label{mezi}
\frac{p_{2n-2}}{q_{2n-2}}<\frac{\lfloor m\gamma\rfloor}{m}<\frac{p_{2n}}{q_{2n}} \quad \text{for some $n\in\mathbb{N}$},
\end{equation}
or $\lfloor m\gamma\rfloor/m$ satisfies
\begin{equation}\label{k nasobek}
\lfloor m\gamma\rfloor=h\cdot p_{2n}\,,\quad m=h\cdot q_{2n} \qquad \text{for some $n\in\mathbb{N}$ and $h\in\mathbb{N}$, $h\geq2$}.
\end{equation}
If \eqref{mezi} is true, Proposition~\ref{Prop. error}~(ii) gives $m|m\gamma-\lfloor m\gamma\rfloor|>1/c_{2n}=1$.
If \eqref{k nasobek} is true, we use Proposition~\ref{Prop. error}~(i), which implies
\begin{equation*}
\begin{split}
m|m\gamma-\lfloor m\gamma\rfloor|&=hq_{2n}|hq_{2n}\gamma-hp_{2n}|=h^2\cdot q_{2n}|q_{2n}\gamma-p_{2n}|>h^2\cdot\frac{1}{c_{2n+1}+\frac{1}{c_{2n+2}}+\frac{1}{c_{2n}}} \\
&\geq h^2\cdot\frac{1}{2+1+1}=\frac{h^2}{4}\geq1.
\end{split}
\end{equation*}
To sum up, $\frac{2}{\pi}m\tan\left(\frac{\pi}{2}(m\gamma-\lfloor m\gamma\rfloor)\right)>m(m\gamma-\lfloor m\gamma\rfloor)>1$ in either case.

\noindent (ii)\; We start again from estimate~\eqref{tan estimate},
$$
\frac{2}{\pi}m\tan\left(\frac{\pi}{2}(m\gamma-\lfloor m\gamma\rfloor)\right)\geq m(m\gamma-\lfloor m\gamma\rfloor)=q_{2n}(q_{2n}\gamma-p_{2n}).
$$
Now we use Proposition~\ref{Prop. error}~(i), which gives
$$
q_{2n}(q_{2n}\gamma-p_{2n})>\frac{1}{c_{2n+1}+\frac{1}{c_{2n+2}+\frac{1}{c_{2n+3}+1}}+\frac{1}{c_{2n}+\frac{1}{c_{2n-1}+1}}}\geq\frac{1}{1+\frac{1}{1+\frac{1}{2+1}}+\frac{1}{1+\frac{1}{2+1}}}=\frac{2}{5}.
$$

\noindent (iii)\; At first, Proposition~\ref{Prop. error}~(i) is used to estimate the quantity
\begin{equation}\label{est.1/3}
q_{2n}(q_{2n}\gamma-\lfloor q_{2n}\gamma\rfloor)=q_{2n}(q_{2n}\gamma-p_{2n})<\frac{1}{c_{2n+1}+\frac{1}{c_{2n+2}+1}+\frac{1}{c_{2n}+1}}\leq\frac{1}{2+\frac{1}{1+1}+\frac{1}{1+1}}=\frac{1}{3}.
\end{equation}
This estimate combined with a trivial inequality $q_{2n}\geq q_2=2$ implies $\gamma-\lfloor q_{2n}\gamma\rfloor<1/6$. Therefore,
$$
\frac{\pi}{2}(q_{2n}\gamma-\lfloor q_{2n}\gamma\rfloor)<\frac{\pi}{2}\cdot\frac{1}{6}=\frac{\pi}{12}.
$$
Since $\tan\frac{\pi}{12}=2-\sqrt{3}$, one has $\tan x\leq\frac{12}{\pi}(2-\sqrt{3})x$ for all $x\in[0,\pi/12)$; hence
\begin{multline*}
\frac{2}{\pi}q_{2n}\tan\left(\frac{\pi}{2}(q_{2n}\gamma-\lfloor q_{2n}\gamma\rfloor)\right)
\leq\frac{2}{\pi}q_{2n}\frac{12}{\pi}(2-\sqrt{3})\cdot\frac{\pi}{2}(q_{2n}\gamma-\lfloor q_{2n}\gamma\rfloor) \\
=\frac{12}{\pi}(2-\sqrt{3})\cdot q_{2n}(q_{2n}\gamma-\lfloor q_{2n}\gamma\rfloor)
<\frac{12}{\pi}(2-\sqrt{3})\cdot\frac{1}{3}
=\frac{4}{\pi}(2-\sqrt{3})\approx0.341,
\end{multline*}
where \eqref{est.1/3} was used.
\end{proof}

We will also need a similar result for the value $\gamma^{-1}$:

\begin{lemma}{Lemma}\label{Levels 1/beta}
Let $\gamma=[0;1,1,c_3,1,c_5,1,c_7,1,\ldots]$, where $c_{2n+1}\in\{1,2\}$ for all $n\in\mathbb{N}$. Then
$$
\frac{2}{\pi}m\tan\left(\frac{\pi}{2}(m\gamma^{-1}-\lfloor m\gamma^{-1}\rfloor)\right)>\frac{1}{3}
$$
for all $m\in\mathbb{N}$.
\end{lemma}

\begin{proof}
One proceeds similarly as in the proof of Lemma~\ref{Levels beta}, applying Proposition~\ref{Prop. error} on $\gamma^{-1}=[1;1,c_3,1,c_5,1,c_7,1,\ldots]$. Let us denote $\gamma^{-1}=[1;c'_1,c'_2,c'_3,c'_4,c'_5,c'_6,c'_7,\ldots]$, i.e., $c'_j=c_{j+1}$ for all $j$.

We start from a trivial estimate
\begin{equation}\label{tg est.}
\frac{2}{\pi}m\tan\left(\frac{\pi}{2}(m\gamma^{-1}-\lfloor m\gamma^{-1}\rfloor)\right)\geq\frac{2}{\pi}m\cdot\frac{\pi}{2}(m\gamma^{-1}-\lfloor m\gamma^{-1}\rfloor)=m(m\gamma^{-1}-\lfloor m\gamma^{-1}\rfloor),
\end{equation}
and then distinguish $3$ possible cases:

\begin{itemize}
\item[(i)] $\lfloor m\gamma^{-1}\rfloor/m$ is a convergent of $\gamma^{-1}$; i.e., $\lfloor m\gamma^{-1}\rfloor=p'_{2n}$, $m=q'_{2n}$ for some $n$. Then Proposition~\ref{Prop. error}~(i) implies
\begin{multline*}
m(m\gamma^{-1}-\lfloor m\gamma^{-1}\rfloor)=q'_{2n}(q'_{2n}\gamma^{-1}-p'_{2n})>\frac{1}{c'_{2n+1}+\frac{1}{c'_{2n+2}}+\frac{1}{c'_{2n}}} \\
=\frac{1}{c_{2n+2}+\frac{1}{c_{2n+3}}+\frac{1}{c_{2n+1}}}\geq\frac{1}{1+\frac{1}{1}+\frac{1}{1}}=\frac{1}{3}.
\end{multline*}

\item[(ii)] $\lfloor m\gamma^{-1}\rfloor=hp'_{2n}$ and $m=hq'_{2n}$ for some $n\in\mathbb{N}$ and $h\geq2$; thus
$$
m|m\gamma^{-1}-\lfloor m\gamma^{-1}\rfloor|=hq'_{2n}|hq'_{2n}\gamma^{-1}-hp'_{2n}|=h^2\cdot q'_{2n}|q'_{2n}\gamma^{-1}-p'_{2n}|.
$$
Proposition~\ref{Prop. error}~(i) allows to estimate $\gamma$ as follows,
$$
h^2\cdot q'_{2n}|q'_{2n}\gamma^{-1}-p'_{2n}|>h^2\cdot\frac{1}{c'_{2n+1}+\frac{1}{c'_{2n+2}}+\frac{1}{c'_{2n}}}\geq h^2\cdot\frac{1}{c_{2n+2}+1+1}=\frac{h^2}{3}\geq\frac{4}{3}\;.
$$

\item[(iii)] $\lfloor m\gamma^{-1}\rfloor/m$ lies between convergents, i.e.,
$$
\frac{p'_{2n-2}}{q'_{2n-2}}<\frac{\lfloor m\gamma^{-1}\rfloor}{m}<\frac{p'_{2n}}{q'_{2n}} \quad \text{for some $n\in\mathbb{N}$}.
$$
Then Proposition~\ref{Prop. error}~(ii) implies $m|m\gamma-\lfloor m\gamma^{-1}\rfloor|>1/c'_{2n}=1/c_{2n+1}\geq1/2$.
\end{itemize}
In all cases the estimate~\eqref{tg est.} gives $\frac{2}{\pi}m\tan\left(\frac{\pi}{2}(m\gamma^{-1}-\lfloor m\gamma^{-1}\rfloor)\right)>\frac{1}{3}$.
\end{proof}

Now we are ready to formulate the main result of this section.
If the parameter $\alpha$ (the strength of the $\delta$ coupling in the vertices) is properly chosen, the number of $2$'s in the continued fraction expansion of $\gamma$ is equal to the number of spectral gaps. Furthermore, the positions of $2$'s directly govern the positions of gaps via the denominators of the convergents of $\gamma$.

\begin{theorem}{Theorem}\label{Gap criterion}
Let $\gamma=[0;1,1,c_3,1,c_5,1,c_7,1,\ldots]$, where $c_{2n+1}\in\{1,2\}$ for all $n\in\mathbb{N}$. Consider a $d$-dimensional quantum lattice graph with edge lengths $a_1=\gamma a$ and $a_j=a$ for $j=2,\ldots,d$, which has a $\delta$-coupling with parameter $\alpha\in\left[\frac{4\pi(2-\sqrt{3})}{a},\frac{2\pi^2}{5a}\right]$ in each vertex.
Then $k^2$ is a lower endpoint of a spectral gap if and only if $k^2=\left(\frac{q_{2n}\pi}{a}\right)^2$ and $c_{2n+1}=2$, where $q_{2n}$ denotes the denominator of the $(2n)$-th convergent of $\gamma$.
\end{theorem}

\begin{proof}
The lattice has two types of edge lengths: $\gamma a$ and $a$. Therefore, due to Theorem~\ref{Thm. gaps F}, the spectral gaps can have two types of the lower endpoints: $\left(\frac{m\pi}{\gamma a}\right)^2$ and $\left(\frac{m\pi}{a}\right)^2$, where $m\in\mathbb{N}$. Let us examine each possibility.

(i)\; A gap with lower endpoint at $\left(\frac{m\pi}{\gamma a}\right)^2$ occurs if and only if
$$
\frac{2m\pi}{\gamma a}\left[\tan\left(\frac{\pi}{2}\left(m\frac{a}{a}-\left\lfloor m\frac{a}{a}\right\rfloor\right)\right)+(d-1)\tan\left(\frac{\pi}{2}\left(m\frac{a}{\gamma a}-\left\lfloor m\frac{a}{\gamma a}\right\rfloor\right)\right)\right]<\alpha,
$$
(see Theorem~\ref{Thm. gaps F}),
which is equivalent to
\begin{equation}\label{gap I}
\frac{2}{\pi}m\tan\left(\frac{\pi}{2}\left(m\gamma^{-1}-\left\lfloor m\gamma^{-1}\right\rfloor\right)\right)<\frac{\gamma}{d-1}\cdot\frac{a\alpha}{\pi^2}\;.
\end{equation}
Using Proposition~\ref{Prop. ineq}, we get
$$
\gamma<[0;1,1,c_3]\leq[0;1,1,1]=\frac{1}{1+\frac{1}{1+\frac{1}{1}}}=\frac{2}{3}\;.
$$
Therefore, the right hand side of \eqref{gap I} satisfies, with regard to $d\geq2$ and the assumption $a\alpha<2\pi^2/5$,
$$
\frac{\gamma}{d-1}\cdot\frac{a\alpha}{\pi^2}<\frac{2}{3}\cdot\frac{2}{5}=\frac{4}{15}\;.
$$
At the same time, according to Lemma~\ref{Levels 1/beta}, the left hand side of \eqref{gap I} is greater than $1/3$. 
Consequently, inequality~\eqref{gap I} never holds true; i.e, there are no spectral gaps having lower endpoints at $\left(\frac{m\pi}{\gamma a}\right)^2$ for $m\in\mathbb{N}$.

(ii)\; A gap with lower endpoint at $\left(\frac{m\pi}{a}\right)^2$ occurs if and only if
$$
\frac{2m\pi}{a}\left[\tan\left(\frac{\pi}{2}\left(m\frac{\gamma a}{a}-\left\lfloor m\frac{\gamma a}{a}\right\rfloor\right)\right)+(d-1)\tan\left(\frac{\pi}{2}\left(m\frac{a}{a}-\left\lfloor m\frac{a}{a}\right\rfloor\right)\right)\right]<\alpha,
$$
which is equivalent to
\begin{equation}\label{gap II}
\frac{2}{\pi}\tan\left(\frac{\pi}{2}\left(m\gamma-\lfloor m\gamma\rfloor\right)\right)<\frac{a\alpha}{\pi^2}.
\end{equation}
Lemma~\ref{Levels beta} gives the following estimates on the left hand side of~\eqref{gap II}:
\begin{equation*}
\begin{array}{lcl}
\mathrm{LHS}<\frac{4}{\pi}(2-\sqrt{3}) &\quad& \text{if $m=q_{2n}$ and $c_{2n+1}=2$}; \\[2pt]
\mathrm{LHS}>\frac{2}{5} &\quad& \text{otherwise}.
\end{array}
\end{equation*}
Regarding the right hand side of~\eqref{gap II}, the assumptions on $\alpha$ imply that
$$
\mathrm{RHS}=\frac{a\alpha}{\pi^2}\in\left[\frac{4}{\pi}(2-\sqrt{3}),\frac{2}{5}\right].
$$
Therefore, inequality~\eqref{gap II}, $\mathrm{LHS}<\mathrm{RHS}$, is satisfied if and only if $m=q_{2n}$ and $c_{2n+1}=2$. In other words, lower endpoints of spectral gaps are exactly the points $\left(\frac{q_{2n}\pi}{a}\right)^2$ such that $c_{2n+1}=2$.
\end{proof}

Theorem~\ref{Gap criterion} gives a method to construct a periodic quantum graph with any prescribed number of gaps, whose positions can be to some extent arranged in a desired manner.
Let us illustrate the result on an example. (N.B. The notation $\bar{1}$ in the continued fraction expansion represents an infinite sequence $1,1,1,\ldots$.)

\begin{example}{Example}\label{Example BS}
Let $a>0$ and $\alpha a\in[4\pi(2-\sqrt{3}),2\pi^2/5]$. Below we apply Theorem~\ref{Gap criterion} on three sample choices of $\gamma$. The number of gaps is given by the number of occurrences of $2$ in the continued fraction expansion of $\gamma$. The gaps begin at points $k^2=\left(\frac{q_{2n}\pi}{a}\right)^2$, where $n$ is given by the condition $c_{2n+1}=2$ and $p_{2n}/q_{2n}$ is the $2n$-th convergent of $\gamma$.
\begin{itemize}
\item If $\gamma=[0;1,1,2,\bar{1}]=(15-\sqrt{5})/22$, the spectrum has only one gap, whose lower endpoint is located at $(q_{2}\pi/a)^2$. Since the second convergent of $\gamma$ is $\frac{p_2}{q_2}=\frac{1}{1+\frac{1}{1}}=\frac{1}{2}$, the lower endpoint of the gap is $k^2=(q_{2}\pi/a)^2=4\pi^2/a^2$.
\item If $\gamma=[0;1,1,1,1,2,\bar{1}]=(99-\sqrt{5})/158$, the spectrum has again only one gap, whose lower endpoint is this time located at $(q_{4}\pi/a)^2$. The fourth convergent of $\gamma$ is $\frac{p_4}{q_4}=\frac{3}{5}$, hence $(q_{4}\pi/a)^2=25\pi^2/a^2$.
\item If $\gamma=[0;1,1,2,1,2,\bar{1}]=(209-\sqrt{5})/358$, the spectrum has exactly $2$ gaps, whose lower endpoints are located at $(q_{2}\pi/a)^2$ and $(q_{4}\pi/a)^2$. Since the $2$nd and the $4$th convergent of $\gamma$ are $\frac{p_2}{q_2}=\frac{1}{2}$ and $\frac{p_4}{q_4}=\frac{4}{7}$, the lower endpoints of the gaps are $4\pi^2/a^2$ and $49\pi^2/a^2$.
\end{itemize}
\end{example}

Note that Example~\ref{Example BS} gives a constructive proof of the existence of quantum graphs that are periodic in $d$ dimensions ($d\geq2$) and have a nonzero finite number of spectral gaps. In case of $d\geq3$ this is the first such example to date.

\begin{remark}{Remark}
In the main theorem of this section (Theorem~\ref{Gap criterion}), we restricted our attention to the special choice of $\gamma$, namely, $\gamma=[0;1,1,c_3,1,c_5,1,c_7,1,\ldots]$, where $c_{2n+1}\in\{1,2\}$. Such restriction is made only for the sake of simplicity, because it allowed us to keep the formulations and proofs of Lemmas~\ref{Levels beta} and \ref{Levels 1/beta} rather simple. But naturally, the idea presented here works generally for any $\gamma$. The procedure would be analogous: One would examine the quantities $q(q\gamma-\lfloor q\gamma\rfloor)$ (separately when $\lfloor q\gamma\rfloor/q$ is a convergent of $\gamma$ and when it is not a convergent), and similarly for $\gamma^{-1}$. Then one would properly estimate the tangent function in the expression $\frac{2}{\pi}q_{2n}\tan\left(\frac{\pi}{2}(q_{2n}\gamma-\lfloor q_{2n}\gamma\rfloor)\right)$ (similarly for $\gamma^{-1}$). In this way one finds ``levels'' that determine, when compared to the quantity $a\alpha/\pi^2$, the gaps in the spectrum. Roughly speaking, the number and positions of gaps are closely related (although not precisely coinciding) to the \emph{number of occurences and the positions of large terms in the continued fraction expansion of $\gamma$}.
\end{remark}

\section{Conclusions}

The results of this paper can be summarized in three points. The first one consists in a characterization of the positions of gaps in the spectrum of a quantum graph taking the form of a $d$-dimensional hyperrectangular lattice for a general $d$ ($d=3$ corresponds to a cuboidal lattice). Secondly, we found sufficient conditions on the parameters of the system guaranteeing that the number of spectral gaps is finite. The criteria rest upon the measure of irrationality of edge lengths ratios, which is expressed in terms of the function $\upsilon(\gamma)$. Thirdly, a connection between the number and positions of spectral gaps and the structure of the continued fraction expansion of the edge lengths ratio was found and explored.
This led to a precise specification of the number and locations of gaps in the spectrum of a lattice with two edge lengths related by ratio $\gamma$. Although we considered only a $\gamma$ having a special continued fraction expansion in order to keep the presentation simple enough, the idea is valid generally and its application to other values of $\gamma$ would be straightforward.

As a by-product, the results provide a constructive proof of the existence of quantum graphs that are periodic in $d$ dimensions for a general $d\geq3$ and have a nonzero finite number of gaps in their spectra. Together with the achievements of \cite{ET17} on graphs periodic in $2$ dimensions, the existence of periodic quantum graphs having a nonzero finite number of spectral gaps is now established in any dimension $d\geq2$.

An open question remains regarding graphs that are periodic in only one dimension, which take the form of an infinite periodic chain (so-called ``$\mathbb{Z}$-periodic graphs''~\cite[Sect.~4.6 and 4.7]{BK13}). It is not known whether a $\mathbb{Z}$-periodic graph having a nonzero finite number of spectral gaps exists. No example has appeared in the literature so far, but from a general result, we know that such a quantum graph $(\Gamma,H)$, if it exists, will have the following two properties \cite{ET17}:
\begin{itemize}
\item Some of the vertex couplings lie outside the scale-invariant class.
\item Let $(\Gamma,H_0)$ be obtained from the graph $(\Gamma,H)$ in question by replacing each vertex coupling with the so-called ``associated scale-invariant couplings'' (i.e., having the Robin component removed; see \cite[Def.~2.4]{ET17} for details). Then $H_0$ must have no gaps in the spectrum.
\end{itemize}

\end{document}